\documentclass[pra,twocolumn,superscriptaddress,letterpaper]{revtex4-2}

\usepackage{nicefrac}
\usepackage{braket}
\usepackage{color}
\usepackage{graphicx,ulem,bbm}
\usepackage[colorlinks=true,urlcolor=blue,citecolor=blue,linkcolor=blue]{hyperref}

\usepackage{amssymb,amsmath,amsfonts}
\usepackage{amsthm}
\newtheorem{theorem}{Theorem}
\newtheorem{corollary}{Corollary}
\newtheorem{lemma}{Lemma}

\newtheorem{remark}{Remark}

\begin{document}

\title{%
   Multiphoton, multimode state classification for nonlinear optical circuits
}

\author{Denis A. Kopylov}
    \email{denis.kopylov@uni-paderborn.de}
    \affiliation{Institute for Photonic Quantum Systems (PhoQS), Paderborn University, Warburger Stra\ss{}e 100, D-33098 Paderborn, Germany}
    \affiliation{Department of Physics, Paderborn University, Warburger Stra\ss{}e 100, D-33098 Paderborn, Germany}

\author{Christian Offen}
    \affiliation{Institute for Photonic Quantum Systems (PhoQS), Paderborn University, Warburger Stra\ss{}e 100, D-33098 Paderborn, Germany}
    \affiliation{Department of Mathematics, Paderborn University, Warburger Stra\ss{}e 100, D-33098 Paderborn, Germany}

\author{Laura Ares}
    \affiliation{Institute for Photonic Quantum Systems (PhoQS), Paderborn University, Warburger Stra\ss{}e 100, D-33098 Paderborn, Germany}
    \affiliation{Department of Physics, Paderborn University, Warburger Stra\ss{}e 100, D-33098 Paderborn, Germany}

\author{Boris Wembe Moafo}
    \affiliation{Department of Mathematics, Paderborn University, Warburger Stra\ss{}e 100, D-33098 Paderborn, Germany}

\author{Sina Ober-Bl\"obaum}
    \affiliation{Department of Mathematics, Paderborn University, Warburger Stra\ss{}e 100, D-33098 Paderborn, Germany}

\author{Torsten Meier}
    \affiliation{Institute for Photonic Quantum Systems (PhoQS), Paderborn University, Warburger Stra\ss{}e 100, D-33098 Paderborn, Germany}
    \affiliation{Department of Physics, Paderborn University, Warburger Stra\ss{}e 100, D-33098 Paderborn, Germany}

\author{Polina R. Sharapova}
    \affiliation{Department of Physics, Paderborn University, Warburger Stra\ss{}e 100, D-33098 Paderborn, Germany}

\author{Jan Sperling}
    \affiliation{Institute for Photonic Quantum Systems (PhoQS), Paderborn University, Warburger Stra\ss{}e 100, D-33098 Paderborn, Germany}
    \affiliation{Department of Physics, Paderborn University, Warburger Stra\ss{}e 100, D-33098 Paderborn, Germany}

\date{\today}

\begin{abstract}
    We introduce a new classification of multimode states with a fixed number of photons.
    This classification is based on the factorizability of homogeneous multivariate polynomials and is invariant under unitary transformations.
    The classes physically correspond to field excitations in terms of single and multiple photons, each of which being in an arbitrary irreducible superposition of quantized modes.
    We further show how the transitions between classes are rendered possible by photon addition, photon subtraction, and photon-projection nonlinearities.
    We explicitly put forward a design for a multilayer interferometer in which the states for different classes can be generated with state-of-the-art experimental techniques.
    Limitations of the proposed designs are analyzed using the introduced classification, providing a benchmark for the robustness of certain states and classes.
\end{abstract}

\maketitle


\section{Introduction}

    A central building block of most photonic quantum technologies is the multimode interferometer, which allows for linear and nonlinear coupling between different field modes.
    Linear interferometers, consisting of beam splitters and phase shifters, are described by unitary maps \cite{Reck_1994}.
    The indistinguishability of photons leads to quantum interference \cite{Hong_1987} between such field excitations and can produce highly nontrivial relations between the input and output states.
    For this reason, general linear photonic quantum simulators are computationally hard to simulate for classical computers \cite{Aaronson__2010,Olson_2018}, resulting in recent boson sampling experiments \cite{Zhong_2018,Wang_2019_prl_BS,Brod_2019,Taballione_2023} and advances in quantum walks \cite{Gr_fe_2016,DiFidio2024}, to name but a few.
    Also, a wide range of challenging graph problems can be solved in this manner, utilizing encodings via interferometers \cite{Mezher_2023}.

    Exceeding unitary transformations, nonlinearities are required, i.e., optical elements providing nonlinear coupling between the different modes of the electromagnetic field.
    The presence of nonlinearities additionally increases the complexity of quantum systems \cite{Steinbrecher_2019, Spagnolo_2023}, as well as allowing us to realize universal photonic quantum computing \cite{Knill_2001_nature, Kok_2007}. 
    Experimentally, significant nonlinear coupling can be achieved via the so-called measurement-based nonlinearities, which comprise photon detection and post-selection protocols \cite{Scheel_2003, Stanisic_2017, Gubarev_2020}.

    For many quantum technological applications, the generation of particular states is paramount.
    Multimode interferometers equipped with photon addition and subtraction offer powerful tools to satisfy such a demand \cite{Sperling_2014,Fan_2018,Biagi_2022}.
    Thus, the fundamental understanding of multimode interferometers and the states that can be reached by such setups is crucial.
    For this reason, the question of the availability of transformation between states becomes a primary concern, and theoretical methods are in demand.
    For linear multimode interferometers, the mutual state transformations have been analytically studied, e.g., in Refs. \cite{Migda__2014,Parellada_2023}, where invariants and necessary criteria have been proposed.
    However, the analysis becomes more challenging when dimensionality increases and the resources provided by nonlinearities are considered.

    The aforementioned challenge is usually overcome by numerical simulations, which have a high computational prize tag caused by the underlying permanent calculation \cite{scheel2004permanents,Aaronson__2010,Heurtel_2023}.
    Here, numerical simulations within the framework of tensor network states \cite{Or_s_2019} can be effective for certain types of states, e.g., matrix-product states \cite{ Vidal_2003, Dhand_2018, Lubasch_2018}, because of their internal structure and by harnessing tools from tensor algebra \cite{Or_s_2014}.
    For arbitrary states, however, this framework does not reduce the scaling problem.
    Yet, it provides a more intuitive understanding of the behavior of the states.

    In this work, we theoretically address the question of state transformations within rather generic multimode nonlinear interferometers, for instance, serving as a means to produce complexly correlated multiphoton states.
    This encompasses the introduction of a classification for the set of $N$-photon states in a $M$-mode system by mapping this problem to the factorization of multivariate polynomials.
    These classes are shown to be invariant under unitary transformation;
    directed transitions between classes are shown to be achieved by introducing additional nonlinearities.
    In particular, we demonstrate how the multilayer nonlinear interferometers can be designed for the generation of states in different classes, starting with vacuum.

    The structure of the paper is as follows:
    In Sec. \ref{sec:fock_general}, we describe the $N$-photon quantum state as multivariate polynomials and introduce the classification based on the factorization of polynomials.
    In Sec. \ref{sec:multivariate_description}, we study the properties of multilayer interferometers with nonlinearities, differentiating between photon additions, subtractions, and projections.
    In Sec. \ref{sec:examples}, we present interesting examples of states generated via multilayer interferometers.
    We summarize general properties of the framework and discuss possible extensions in Sec. \ref{sec:examples}.
    We conclude in Sec. \ref{sec:conclusions}.
    Appendices provide additional context.

\section{Classification of multimode, multiphoton states}
\label{sec:fock_general}

    In this section, we introduce a classification for $N$-photon states in $M$-mode systems.
    Mapping these states to multivariate homogeneous polynomials allows us to introduce the factorized form of the $N$-photon states in terms of creation operators.
    The uniqueness of the factorized form then provides a natural classification.

    \subsection{Quantum states with a total of $N$ excitations across $M$ bosonic modes}

        Let $\Phi_{N,M}$ denote the set of all $N$-photon states in an $M$-mode system.
        Considering the orthonormal photon-number states, the following set represents a basis: $B_{N,M}=\{|s_1, \ldots, s_M\rangle:s_1,\ldots,s_M\in\mathbb N,s_1+\cdots+s_M = N\}$.
        Then, each state in $\Phi_{M,N}$ has the following, generic form
        \begin{equation} 
            \label{eq:standard_state_NM}
            \ket{\Psi} = \sum_{|s_1,\ldots,s_M\rangle\in B_{N,M} }
            f_{s_1,\ldots,s_M}
            \ket{s_1, \ldots, s_M},
        \end{equation}
        obeying $\langle\Psi|\Psi\rangle=\sum_{s_1, \ldots, s_M}  |f_{s_1, \ldots, s_M}|^2 = 1$, where  $f_{s_1, \ldots, s_M}$ are probability amplitudes.
        The number of basis states, thus the dimension of the Hilbert space $\mathcal H_{N,M}=\mathrm{span}(B_{N,M})$, is given by $\dim\mathcal{H}_{N,M} = {\binom {M+N-1}{N}}$.

        In an $M$-port, linear interferometer, the input creation operators $\boldsymbol{a}^\dagger $ are mapped to the output operators $\boldsymbol{b}^\dagger $ by an $M \times M$ unitary matrix $U$,
        \begin{equation}
            \label{eq:lin_interfer}
            \boldsymbol{b}^\dagger = U \boldsymbol{a}^\dagger.
        \end{equation}
        The input and output states are connected via a unitary evolution operator $\hat{\mathcal{U}}$ that acts on the full Hilbert space $\mathcal H$, making the problem of finding the output state non-trivial.
        Indeed, for two basis states, $\ket{T}=\ket{t_1, \cdots, t_M}$ and $\ket{S} = \ket{s_1, \cdots, s_M}$, the transition matrix elements can be calculated using the permanent $\mathrm{perm}(U_{S,T}) $.
        That is, we have $\bra{T} \hat{\mathcal{U}} \ket{S} = \mathrm{perm}(U_{S,T}) [s_1!\ldots s_M!t_1!\ldots t_M!]^{-1/2}  $, where the $N\times N$ matrix $U_{S,T}$ is constructed from the unitary matrix  $U$ by repeating $s_i$ times its $i$\textsuperscript{th}  column and $t_j$ times its $j$\textsuperscript{th} row;
        for details, see, e.g., \cite{scheel2004permanents,Heurtel_2023}.
        In general, having an arbitrary initial state $\ket{\psi}_\mathrm{in}$, the calculation of the output state $\ket{\psi}_\mathrm{out}$ requires the computation of all matrix elements $\bra{T} \hat{\mathcal{U}} \ket{S} $, and thus the computation of  $(\dim\mathcal{H}_{M,N})^2$ permanents.

    \subsection{Polynomial representation and irreducible factorization thereof}
    \label{sec:factorizedform}

        States in Eq. \eqref{eq:standard_state_NM} can be recast in terms of the creation operators $\hat{a}_i^\dagger$,
        \begin{equation}
        \begin{aligned}
            \label{eq:unfactorized_state}
            \ket{\Psi} = \sum_{\substack{s_1,\ldots, s_M\in\mathbb N: \\ s_1+\cdots+s_M = N} }
            {}&
            \frac{f_{s_1, \ldots, s_M}}{\sqrt{s_1!\cdots s_M!}}
            \\
            {}&
            \times
            \hat a_1^{\dag\,s_1}\cdots\hat a_M^{\dag\,s_M}|\mathrm{vac}\rangle,
        \end{aligned}
        \end{equation}
        where $|\mathrm{vac}\rangle=|0,\ldots,0\rangle$.
        In this form, we can encode the quantum state $\ket{\Psi} $ via a degree-$N$, homogeneous, multivariate polynomial as $|\Psi\rangle=p_\Psi(\hat a_1^\dag,\ldots,\hat a_M^\dag)|\mathrm{vac}\rangle$,
        where $\boldsymbol x=(x_1,\ldots,x_M)$ for $M$ arguments,
        \begin{equation}
            p_\Psi(\boldsymbol x)
            =
            \sum_{\substack{s_1,\ldots, s_M\in\mathbb N: \\ s_1+\cdots+s_M = N} }
            \frac{f_{s_1, \ldots, s_M}}{\sqrt{s_1!\cdots s_M!}}
            x_1^{s_1}\cdots x_M^{s_M},
        \end{equation}
        and where all terms have the same, fixed degree $N=s_1+\cdots+s_M$.
        Moreover, since the creation operators for different modes commute, the elements of $\Phi_{N,M}$ can be identified with elements of the commutative ring of multivariate polynomials, $\mathbb{C}[x_1, ..., x_M]$.

        For each homogeneous polynomial $ p_{\Psi}(\boldsymbol{x})$ of degree $N$ exists a factorization of the following form \cite{Sharpe_1987}:
        \begin{equation}
          p_{\Psi}(\boldsymbol{x}) = \mathcal{N} \prod_{i} \prod_{k_i} p_{k_i}^{(i)}(\boldsymbol{x}),
        \end{equation}
        where the factors  $p_{k_i}^{(i)}(\boldsymbol{x})$ are homogeneous and irreducible polynomials  of $i$\textsubscript{th} order and $\mathcal{N}\in\mathbb C$ is a constant.
        In contrast to univariate polynomials over the complex field, $\mathbb C[x_1]$, which can be always decomposed in terms of linear factors $p_k^{(1)}$, in the multivariate case, $M>1$, the polynomial $p_\Psi(\boldsymbol{x})$ of degree $N \geq 2$ can be irreducible itself.
        Specifically, the factorized form of the aforementioned quantum state is
        \begin{equation}
            \label{eq:factorized_state}
        \begin{aligned}
             \ket{\Psi}
             = {}&
             \mathcal{N}
             \prod_{i=1}^{n_1} \left(\sum_{j_1=1}^M \gamma^i_{j_1} \hat{a}^\dagger_{j_1}\right)
             \prod_{i=1}^{n_2}\left(\sum_{j_1,j_2=1}^M\gamma^i_{j_1,j_2} \hat{a}^\dagger_{j_1}\hat{a}^\dagger_{j_2} \right)
             \\
             &\times
             \prod_{i=1}^{n_3}\left(\sum_{j_1,j_2,j_3=1}^M\gamma^i_{j_1,j_2,j_3} \hat{a}^\dagger_{j_1} \hat{a}^\dagger_{j_2} \hat{a}^\dagger_{j_3} \right)\dots\ket{\mathrm{vac}},
        \end{aligned}
        \end{equation}
        where each bracket corresponds to an irreducible homogeneous polynomial $p_k^{(i)}$, and where  $n_1$ is the number of linear $p_k^{(1)}$, $n_2$ is the number of quadratic, etc. polynomials, with $1n_1+2n_2+3n_3+\cdots=N$ to result in the sought-after total degree $N$.

        Expressing the state in Eq. \eqref{eq:unfactorized_state} in the factorized form of  Eq. \eqref{eq:factorized_state} requires the factorization of a high-dimensional multivariate polynomial, which is quite challenging in practice.
        Moreover, any small perturbation in the probability amplitudes might lead to an irreducible polynomial, rendering numerical computations challenging.
        Nevertheless, there are robust and stable numerical algorithms that allow us to factorize the polynomials with a given precision \cite{Kaltofen_2008,Wu_2015}.
        Specifically, for polynomials over algebraic number fields, the factorization algorithm is realized in the SageMath package \cite{SageMath}.

    \subsection{Classification and discussion}
    \label{sec:classification}

        It is well established that the symmetric tensor algebra $S(\mathbb{C}^M)$ and polynomials in $M$ variables $\mathbb{C}[x_1,\ldots,x_M]$ are isomorphic \cite[pp. 506]{Bourbaki1989}.
        As a consequence, the multiplication in the polynomial ring corresponds to the symmetrized tensor product in $S(\mathbb{C}^M)$ that yields the bosonic exchange symmetry of quantum states.
        Therefore, factorizations of polynomials correspond to symmetric tensor decompositions. 
        In other words, with respect to the symmetric tensor product, we can ask ourselves to which extend can we factorize the photons' states as joint excitation of multiple quantized modes, hence addressing the bosonic entanglement of symmetric particle states \cite{Jan_Agudelo_23}.

        As mentioned before, the non-negative integer $n_k$ in Eq. \eqref{eq:factorized_state} describes how many irreducible degree-$k$ polynomials there are to factorize the state,
        and the collection $(n_1,\ldots,n_N)$ is unique for a given $\ket{\Psi}$;
        see Lemma \ref{lemma_1} in Appendix \ref{appendix_cl_proof}.
        This allows us to introduce a classification scheme of the quantum states in the set $\Phi_{N,M}$ based on the polynomial factorization of $p_\Psi(\boldsymbol x)$, where
        the different classes encompass the states with different tuples $(n_1, n_2, \dots,n_N)$.

        The previously described kind of factorization allows us to allocate processes of collective excitations of different orders.
        Indeed, the integer number $n_k$ shows how many simultaneous excitations of $k$ photons across $M$ modes.
        For instance, $n_1=N$ means that each of the $N$ photons is excited strictly independently but may be generated in various superpositions of modes, $\gamma^i_{1}\hat a_1^\dag+\cdots+\gamma^i_{M}\hat a_M^\dag$. 
        When other values of $n_k\neq n_1$ are non-zero, multiple photon excitations are necessary to create the state.
        For example, if $n_2>0$, there are at least two photons that must have been generated simultaneously as a photon pair;
        likewise, they are the result of a joint irreducible two-photon excitation, $\gamma^i_{1,1}\hat a_1^{\dag2}+\cdots+\gamma^{i}_{j_1,j_2}\hat a_{j_1}^\dag\hat a_{j_2}^\dag+\cdots+\gamma^i_{M,M}\hat a_M^{\dag 2}$, displaying the correlations in this photon pair.

        To designate each class, we utilize the following notation: $[1^{n_1} ~ 2^{n_2} ~ ...~N^{n_N}]_M$, i.e.,
        \begin{equation}
            |\Psi\rangle\in[1^{n_1} ~ 2^{n_2} ~ ...~N^{n_N}]_M
            \Leftrightarrow
            |\Psi\rangle \text{ as in Eq. \eqref{eq:factorized_state}},
        \end{equation}
        listing factors as ``[$\ldots$ degree\textsuperscript{how many}$\ldots$]\textsubscript{number of modes}.''
        For example, the class $[1^{N}]_M$ corresponds to the states that can be linearly factorized, and the class $[N^{1}]_M$ determines states given by one $N$\textsuperscript{th}-degree irreducible polynomial.
        In between, as another example, the class $[1^{4} ~ 2^{0} ~ 3^{2}]_5 \equiv [1^{4} ~ 3^{2}]_5$  includes all the states in a five-port interferometer whose factorization in Eq. \eqref{eq:factorized_state} contains four first-order and two third-order polynomials;
        the number of photons in this case is $1\times4+3\times2=10$.

\begin{figure}
    \includegraphics[width=1.\linewidth]{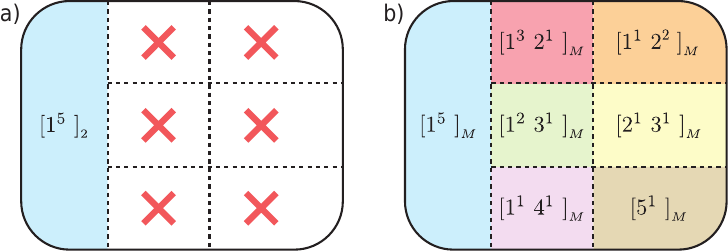}
    \caption{%
        Scheme of all possible classes of $5$-photon states for (a) $M=2$ and (b) $M\geq 3$ bosonic modes.
    }\label{fig:scheme_classes}
\end{figure}

        For two-port interferometers ($M=2$) and a given $N$, there exists only the one class $[1^N]_2$ as any two-variable homogeneous polynomial can be factorized into linear polynomials.
        Non-trivial factorizations emerge for higher-mode systems, $M\geq 3$, where the number of classes corresponds to the number of integer partitions;
        see Theorem \ref{theorem_1} in Appendix \ref{appendix_cl_proof}.
        In Fig. \ref{fig:scheme_classes}, all the possible classes of states with $N=5$ photons are represented for $M=2$ and $M\geq 3$ systems as examples.

        The important property of the factorization in Eq. \eqref{eq:factorized_state} is the fact that the tuple $(n_1, n_2, ... , n_N)$ remains unchanged when the variables are transformed through unitary transformations;
        see Theorem \ref{theorem_2} in Appendix \ref{appendix_cl_proof}.
        This means that a linear optical interferometer cannot change the class of the state, showing a robustness against the choice of mode basis \cite{Sperling_2019}.
        This also emphasizes the importance of nonlinear networks for generating and transitioning to more complex quantum correlations.
        Such nonlinear interferometers and their optical implementation are characterized in the remainder of this work.


\section{Multilayer interferometer}
\label{sec:multivariate_description}

    We now introduce nonlinearities that allow transitions between the classes.
    The first nonlinearity utilized is photon addition.
    Based on this nonlinearity, we present a multilayer interferometer with alternating unitary transformations and photon additions.
    We show that this scheme provides the generation of any state in the class $[1^N]_M$.
    In order to reach a different class, we then introduce photon-subtraction and projection nonlinearities to our multilayer interferometer.

\begin{figure}
    \includegraphics[width=1.\linewidth]{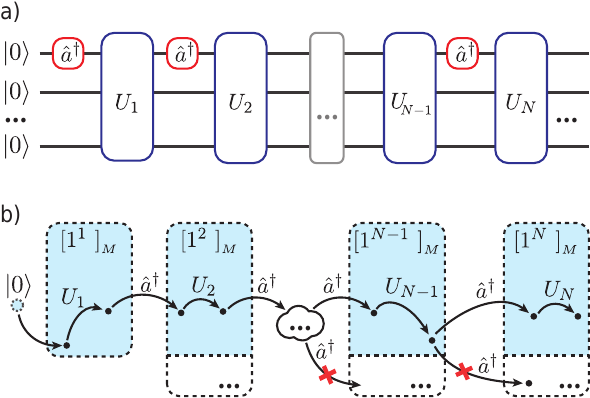}
    \caption{%
        (a) Scheme of the multilayer interferometer with photon additions $\hat{a}^\dagger$ in the first channel and linear transformations $U^{(1)},\ldots,U^{(N)}$.
        (b) Scheme of the transitions between states enabled by the multilayer system.
        Allowed transitions are indicated by arrows.
        Blue regions correspond to the classes $[1^n]_M$, while the white region indicates other classes. 
    }\label{fig:scheme_addition_multilayer}
\end{figure}

    \subsection{Photon-addition nonlinearity}
    \label{sec:nonlinear_additions}

        The first nonlinearity we study is a photon addition, written as $\hat{a}^\dagger$ for a generic mode.
        Clearly, this nonlinear operation increases the number of photons in the system, rendering it possible to utilize the vacuum $|\mathrm{vac}\rangle$ as the input state for all modes of the interferometer.
        Consecutive actions of $N$ photon-creation operators generate $M$-mode basis states,
        \begin{equation}
            \hat{a}_1^{\dag \, s_1}\cdots \hat{a}_M^{\dag \, s_M} \ket{\mathrm{vac}}
            \propto
            \ket{s_1,s_2, ...s_M}\in B_{N,M}.
        \end{equation}

        Going beyond relatively simple basis states, we consider a multilayer system where each photon addition is followed by a unitary transformation;
        see Fig. \ref{fig:scheme_addition_multilayer}(a).
        In the first layer, the state after photon addition reads $\hat{a}_1^\dagger \ket{\mathrm{vac}}$, and the unitary transformation $\hat{a}_i^\dagger \mapsto \sum_{j} U^{(1)}_{i,j} \hat{a}_j^\dagger$ yields a single photon $\ket{\psi}_{1} = \sum_{j} U^{(1)}_{1,j} \hat{a}_j^\dagger \ket{0}$ that is spread across modes.
        The addition of a second photon gives $\hat{a}_1^\dagger \ket{\psi}_{1} = \hat{a}_1^\dagger  \sum_{j} U^{(1)}_{1,j} \hat{a}_j^\dagger \ket{\mathrm{vac}} $, which is followed by a second unitary transformation, resulting in $\ket{\psi}_{2} \propto \left( \sum_{i} U^{(2)}_{1,i} \hat{a}_i^\dagger \right)\Big(   \sum_{j,k} U^{(1)}_{1,j} U^{(2)}_{j,k} \hat{a}_k^\dagger \Big)\ket{\mathrm{vac}} $.
        Iterating this procedure for $N$ layers, we eventually arrive at a state in the following form:
        \begin{equation}
            \label{eq:state_in_c}
            \ket{\Psi} = \mathcal{N} \prod_{i=1}^{N} \left(\sum_{j=1}^M \gamma^{i}_{j} \hat{a}^\dagger_{j}\right) \ket{\mathrm{vac}}
            =
            \mathcal N \hat c_1^\dag \cdots \hat c_N^\dag|\mathrm{vac}\rangle,
        \end{equation}
        which belongs to $[1^N]_M$ since it only includes first-order polynomials.
        Also, we introduced the operators $\hat{c}^\dagger_i = \sum_{j=1}^M \gamma^i_{j} \hat{a}^\dagger_j$ to describe the superposition of modes that is excited by each added photon.
        Each operator $\hat{c}^\dagger_i$ is determined by the row-vector $\Gamma^{[i]} = ( \gamma^i_{1}, \ldots , \gamma^i_{M}) $;
        therefore, the full state $\ket{\Psi}$ can be represented by a multiset (aka unordered tuple) of vectors as
        \begin{equation}
            \label{eq:multiset_vectors}
            \Gamma_\Psi = \left\{\Gamma^{[1]}, \ldots, \Gamma^{[N]}   \right\}.
        \end{equation}

        We note that the polynomial $p_\Psi(\boldsymbol x)=p_1(\boldsymbol x)\cdots p_N(\boldsymbol x)$ that describes the state $\ket{\Psi}$ is not unique in the following sense:
        The multiplication of each linear factor $p_i(\boldsymbol x)$ by a non-zero scalar $s_i$ does not change the physical state.
        In terms of the multiset $\Gamma_\Psi $, this means that the same state is described by $\Gamma^\prime_\Psi = \{s_1 \Gamma^{[1]}, \ldots, s_N\Gamma^{[N]}  \} $.
        If we additionally choose $\Gamma^{[n]}$ to be normalized, the given mutliset $\Gamma_\Psi$ uniquely determines the quantum state, providing the general parametrization of the class $[1^N]_M$.

        Since the state in Eq. \eqref{eq:state_in_c} is the same as a linearly factorized state in Eq. \eqref{eq:factorized_state}, states in classes other than $[1^N]_M$ are not reachable with the multilayer scheme in Fig. \ref{fig:scheme_addition_multilayer}.
        However, can this scheme with photon additions and unitary transformations reach any state in the class $[1^N]_M$?
        The answer is yes and shown in the following.

    \paragraph{Unitary transformations in $[1^N]_M$.}

        Firstly, we analyze a unitary transformation $U$.
        Applying $U$ to the state, the superposition of creation operators $\hat{c}^\dagger_i$ is mapped as
        \begin{equation}
            \label{eq:transform_unitary_operators}
            \hat{c}^\dagger_i
            \mapsto
            \sum_{j=1}^M \gamma^i_{j} \sum_{k=1}^M U_{j,k}  \hat{a}_k^\dagger = \sum_{k=1}^M  \tilde{\gamma}^i_{k} \hat{a}_k^\dagger,
        \end{equation}
        which results in the first-order polynomials, meaning that a unitary transformation may not alter the class.
        The vector $\Gamma^{[i]}$ maps as
        \begin{equation}
             \label{eq:transform_unitary_vectors}
             \Gamma^{[i]}
             \mapsto
             \Gamma^{[i]} U,
        \end{equation}
        with the updated multiset $\Gamma_\Phi = \{\Gamma^{[1]}U, \ldots, \Gamma^{[N]}U \}$ that determines the transformed state.
        Because the unitary $U$ is the same for each $i$, the angles between the vectors $\Gamma^{[n]}$ remain invariant along the linear optical circuit.
        For two-mode states ($M=2$), this result was shown in Ref. \cite{Migda__2014} and discussed as the invariance of Majorana stellar representation under unitary maps.
        Note that the states described by Eq. \eqref{eq:state_in_c} belong to the class $[1^N]_M$ and, therefore, can be interpreted as states defined through an $M$-dimensional Majorana stellar representation.

    \paragraph{Photon additions to $[1^N]_M$.}

        Now, we analyze photon addition in our multiset representation, beginning with a state $\ket{\phi}$ from the class $[1^N]_M$ that is represented by the multiset $\Gamma_\phi = \{ \Gamma^{[1]}, \ldots , \Gamma^{[N]} \}$.
        After a photon addition, the state reads $\ket{\psi} \propto \hat{a}_1^\dagger \ket{\phi}$, which is in $[1^{N+1}]_M$;
        see Fig. \ref{fig:scheme_addition_orbits}.
        Thus, the addition of one photon adds the vector $\Gamma^{[N+1]} = (1, 0, \ldots, 0) $ to the already existing multiset of vectors,
        \begin{equation}
            \label{eq:set_after_photon_addition}
            \Gamma_\psi = \{ \Gamma^{[1]}, \ldots , \Gamma^{[N]}, \Gamma^{[N+1]} \}
        \end{equation}
        If---before the photon addition---we rotate the state $\ket{\phi}$ to $\ket{\tilde{\phi}}$ via a chosen $U$, we therefore rotate the multiset $\Gamma_{\tilde\phi} = \{\tilde\Gamma^{[1]}, \ldots ,\tilde\Gamma^{[N]}\}$ with the rule $\tilde\Gamma^{[j]} = \Gamma^{[j]}U$ as discussed before. 
        The addition then generates a different state $\ket{\tilde{\psi}} \propto \hat{a}_1^\dagger \ket{\tilde{\phi}}$ with vectors $\Gamma_{\tilde\psi} = \{ \tilde\Gamma^{[1]}, ..., \tilde\Gamma^{[N]}, \Gamma^{[N+1]} \}$.
        Now, the angles between the vectors in the multisets $\Gamma_{\tilde\psi}$ and $\Gamma_{\psi}$ are altered according to our choice of $U$;
        thus,  in general, the states $|\psi\rangle$ and $|\tilde\psi\rangle$ cannot be transformed into each other via a unitary transformation;
        see Fig. \ref{fig:scheme_addition_orbits}.

\begin{figure}
    \includegraphics[width=1.0\linewidth]{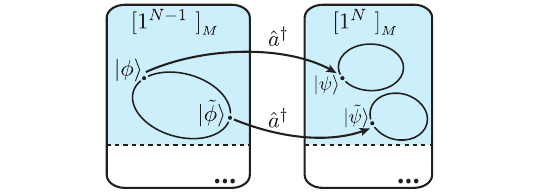}
    \caption{%
        Effect of unitary transformation and photon addition.
        The bold dots represent the states that are connected via unitary transformation.
        After photon addition, the states $\ket{\phi}$ and $\ket{\tilde{\phi}}= U\ket{\phi}$ are transformed into $\ket{\psi} \propto \hat{a}^\dagger\ket{\phi}$ and $\ket{\tilde\psi} \propto \hat{a}^\dagger\ket{\tilde{\phi}}$, respectively.
        The resulting states $\ket{\psi}$ and $\ket{\tilde\psi}$ are not necessarily connected via a unitary transformation anymore.
    }\label{fig:scheme_addition_orbits}
\end{figure}

    \paragraph{Summary.}

        In the multilayer interferometer with alternating unitary transformations and photon additions enables one to construct an arbitrary multiset of normalized vectors, with arbitrary relations of angles between them.
        Therefore, the scheme in Fig. \ref{fig:scheme_addition_multilayer}(a) is universal in the sense that we can reach the entire class $[1^N ]_M$.
        Some concrete examples are presented in Sec. \ref{sec:examples}.

    \subsection{Photon subtractions and photon-projection nonlinearities}

    As described in Sec. \ref{sec:nonlinear_additions}, photon additions alone are insufficient for exceeding the class $[1^N]_M$.
    However, starting from those states, we can go backwards---i.e., apply nonlinearities that reduce the number of photons---in order to approach different classes with complex correlations.
    To this end, we study single-photon subtraction and single-photon projection as experimentally accessible, measurement-induced nonlinearities in this section.
    Using the setup in Fig. \ref{fig:scheme_addition_multilayer} in combination with the nonlinearities discussed here then renders it possible to go from the base class $[1^N]_M$ to more sophisticated classes $[1^{n_1} ~ 2^{n_2} ~ ...~ N^{n_{N'}}]_M$.
    
    \paragraph{Single-photon subtraction.}

        The scheme of the subtraction process is shown in Fig. \ref{fig:scheme_subtraction_projection}(a).
        Applying an annihilation operator that implements the subtraction in the $k$-th mode transforms the state from $\ket{\psi}$ in Eq. \eqref{eq:state_in_c} to
        \begin{equation}
            \ket{\phi} \propto \left(\gamma^1_{k} \hat{c}^\dagger_2 \cdots \hat{c}^\dagger_N + \cdots + \gamma^N_{k}\hat{c}^\dagger_1 \cdots \hat{c}^\dagger_{N-1} \right) \ket{\mathrm{vac}}
        \end{equation}
        because of $\hat a_k\hat a_j^\dag=\delta_{j,k}+\hat a_j^\dag\hat a_k$, which yields the identity
        \begin{equation}
            \hat{a}_k \hat{c}^\dagger_i = \sum_{j=1}^M \gamma^i_{j}  \hat{a}_k \hat{a}^\dagger_j  =   \gamma^i_{k}  + \hat{c}^\dagger_i \hat{a}_k.
        \end{equation}
        More compactly, the state after subtracting one photon in the $k$\textsuperscript{th} mode takes the form
         \begin{equation}
            \label{eq:multivar_state_subtract}
            \ket{\phi} \propto
            \sum_{i=1}^N \gamma^i_{k}  \left( \prod^N_{\substack{j=1 \\ j\neq i}}  \hat{c}^\dagger_j \right)
            \ket{\mathrm{vac}}.
        \end{equation}
        This resulting state has one photon less and is written as a sum---not product---of linearly factorized polynomials and thus is not in the base class $[1^{N-1}]_{M}$, in general.

    \paragraph{Single-photon projective measurement.}
    
        As an alternative to photon subtraction, we can reduce the number of photons of the state by one with a measurement in the $k$\textsuperscript{th} mode.
        This is described through a projection operator $\hat{\Pi}_k = \hat{a}^\dagger_k \ket{\mathrm{vac}}_k\bra{\mathrm{vac}}_k\hat{a}_k$;
        see Fig.~\ref{fig:scheme_subtraction_projection}(b).
        The detection of one photon in the mode $k$ of the state $\ket{\psi}$, Eq. \eqref{eq:state_in_c}, generates the collapsed state
        \begin{equation}
            \label{eq:multivar_state_project}
             \ket{\phi} \propto
             \hat{\Pi}_k  \ket{\psi}  \propto
             \sum_{i=1}^N \gamma^i_{k}  \left(  \prod^N_{\substack{j=1 \\j\neq i}}  \hat{d}^\dagger_j \right)\ket{\mathrm{vac}},
        \end{equation}
        where $\hat{d}_j^\dagger = \hat{c}_j^\dagger - \gamma^j_{k} \hat{a}^\dagger_k $.
        Note that, in contrast to photon subtraction, the projection $\hat{\Pi}_k$ removes the $k$\textsuperscript{th} mode, and the collapsed state has a reduced modes number $M-1$.
        But, analogously to subtraction, the state after projection in Eq. \eqref{eq:multivar_state_project} is, in general, no longer in the class $[1^{N-1}]_{M-1}$ of a purely linear factorization.

    \paragraph{Unitary transformation.}

        The states in Eqs. \eqref{eq:multivar_state_subtract} and \eqref{eq:multivar_state_project} can be combined with another linear transformation $U$, cf. Sec. \ref{sec:nonlinear_additions}.
        Indeed, the operators $\hat{c}^\dagger_j$ and $\hat{d}^\dagger_j$ are straightforwardly transformed with the expression in Eq. \eqref{eq:transform_unitary_operators}.
        The structure of the resulting states remains the same---a sum of linearly factorized polynomials.
        Therefore, a unitary transformation cannot change any class;
        see Theorem \ref{theorem_2} in Appendix \ref{appendix_cl_proof}.

    \paragraph{Remarks.}

        Beyond the explicit procedures presented here, we can certainly also concatenate the described processes, leaving us not only with linear combinations of factorized polynomials but multilinear generalizations thereof, as exemplified later in this work.
        Furthermore, multi-photon subtractions and projective measurements on more than one photon would have a similar impact.
        For the sake of simplicity, however, we forego the explicit consideration of nonlinear multi-photon processes within this work.

\begin{figure}
    \includegraphics[width=1.\linewidth]{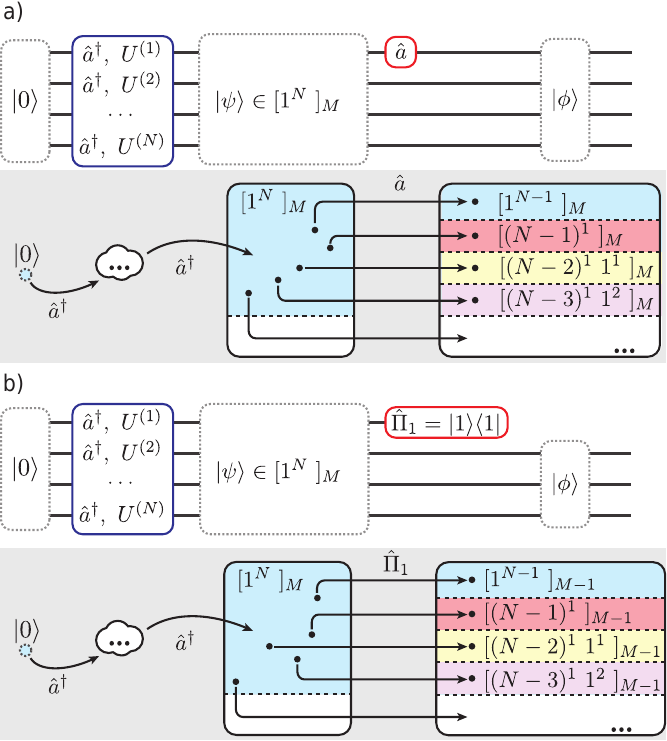} 
    \caption{%
        Scheme of the multilinear interferometer with $N$-additions and rotations together with (a) a single subtraction and (b) a single projection.
        Gray insets represent the transitions between classes under subtraction and projection.
        Different colors represent different classes, and allowed transitions are indicated by arrows.
    }\label{fig:scheme_subtraction_projection}
\end{figure}

    \subsection{Probability amplitudes}
    \label{sec:prob_amplt}

    \paragraph{States in the class $[1^N]_M$.}

        Equation \eqref{eq:factorized_state} provides a straightforward way to compute probability amplitudes.
        More directly, states in the class $[1^N]_M$ are described by linearly factorized polynomials.
        This renders it possible to determine probability amplitudes in terms of permanents as outlined by the following steps.

        In step 1, we stack the row-vectors $\Gamma^{[n]}$ and build the matrix $\Gamma^M$ of size $M\times N$.
        In step 2, for the basis state $\ket{S} = \ket{s_1,\ldots,s_M }$, we prepare the matrix $\tilde \Gamma_S$ from $\Gamma^M$ by taking  $s_1$ times the first column of $\Gamma^M$,  $s_2$ times the second column, etc.
        The resulting matrix $\tilde \Gamma_S$ of size  $N\times N$ now allows us to calculate the unnormalized probability amplitude for the chosen Fock state $\ket{S}$ as
        \begin{equation}
            \label{eq:permanent}
            \tilde{f}_{s_1,\ldots,s_M } = \frac{\mathrm{perm}(\tilde \Gamma_S)}{\sqrt{s_1!\cdots s_M!}}.
        \end{equation}
        The normalized probability amplitudes $f_{s_1,\ldots,s_M }$, obeying the normalization condition $\sum_{s_1,\dots,s_M }  |f_{s_1,\dots,s_M }|^2 = 1$, can be determined after computing $\tilde{f}_{s_1,\dots,s_M }$ for all Fock states.
        Consequently, for $[1^N]_M$ class elements, finding all probability amplitudes requires the computation of only $\dim\mathcal{H}_{N,M}$ permanents.
        Note that for creating $\tilde{\Gamma}_S$ with a structure of repeated columns, more efficient algorithms from Refs. \cite{Shchesnovich_2020,Heurtel_2023} may be used.

    \paragraph{States after subtraction and projection}

        The states after a single subtraction and after a projection are given by the sum of linearly factorized polynomials, cf. Eqs. \eqref{eq:multivar_state_subtract} and \eqref{eq:multivar_state_project}, respectively.
        For the $i$\textsuperscript{th} linearly factorized polynomial in that sum and the Fock state $\ket{S}$, the unnormalized probability amplitudes $ \tilde{f}^{i}_S$ can be computed using Eq. \eqref{eq:permanent}.
        The total unnormalized probability amplitudes then are
        \begin{equation}
            \tilde{f}_S = \sum_i \tilde{f}^{i}_S.
        \end{equation}
        An appropriate normalization can be carried out after computing of all the contributions $\tilde{f}_S$, like done before.


\section{States generation: examples}
\label{sec:examples}

    In this section, we apply the construction schemes devised in the previous section.
    This includes physically relevant base-class states in $[1^N]_M$ as well as highly structured and irreducible states in $[N^1]_M$.
    Furthermore, numerical optimizations to produce sought-after target states are carried out with multiple instances of subtractions and projections, utilizing and generalizing the single-instance description as derived above.

    \subsection{Base class $[1^N]_M$}

        Already the states in the base class $[1^N]_M$ appear in different, pioneering topics in quantum optics and quantum technologies, such as boson sampling \cite{Brod_2019} and quantum walks \cite{Gr_fe_2016}.
        Also, the class $[1^N]_M$ further includes interesting elements, such as states that remain mode-entangled under unitary transformations \cite{Sperling_2019}.
        As developed above, an arbitrary state in the base class can be generated via the multilayer interferometer in Fig. \ref{fig:scheme_addition_multilayer}(a).
        Therefore, the remaining task of generating a specific state is reduced to finding the unitary matrices $U^{(n)}$.
        This can be done, for example, via numerical optimization schemes.

    \paragraph{Symmetric example for a two-mode system.}

        As mentioned in Sec. \ref{sec:factorizedform}, all $N$-photon states in the two-mode system are members of the $[1^N]_{M=2}$ class.
        As one example, one can generate NOON states \cite{Kok_2002,Pryde_2003}, $ \ket{\mathrm{NOON}} = \left( \ket{N,0} + \ket{0,N} \right)/\sqrt2$, which can be defined in terms of the polynomials as $x_1^N + x_2^N = \prod_{n=1}^N (x_1 + \omega_N^n x_2)$, where $\omega_N=\exp(2\pi i/N)$.
        Thus, we have
        \begin{equation}
            \ket{\mathrm{NOON}} \propto \prod_{n=1}^N \left(\hat{a}_1^\dagger + \omega_N^n\hat{a}_2^\dagger\right) \ket{\mathrm{vac}},
        \end{equation}
        and the multiset $\Gamma_{\mathrm{NOON}}$ consists of $N$ vectors, $\Gamma^{[n]} = (1,  \omega_N^n)$ for $1\leq n\leq N$.
        The unitary matrices needed to generate NOON states with the multilayer interferometer in Fig. \ref{fig:scheme_addition_multilayer} (a) are
        \begin{equation}
            U^{(n)}
            = \begin{pmatrix}
                \cos(\pi/N)   & -i\sin(\pi/N)
                \\
                -i\sin(\pi/N)   &  \cos(\pi/N)
            \end{pmatrix},
        \end{equation}
        for $1\leq n <N$, and
        \begin{equation}
            U^{(N)} =
            \begin{pmatrix}
                \frac1{\sqrt{2}}    & -\frac1{\sqrt{2}}
                \\
                \frac1{\sqrt{2}}    & \frac1{\sqrt{2}}
            \end{pmatrix}.
        \end{equation}
        This result is consistent with the decomposition found in Ref. \cite{Kok_2011}.
        Further note that the choice of matrices $U^{(n)}$ is not unique.

    \paragraph{Asymmetric example for a three-mode system.}

        Another interesting state which belongs to the linearly factorizable class is the antisymmetric state \cite{Jex_2003}
        \begin{equation}
        \begin{aligned}
            \ket{\Psi} = {}&
            \frac{1}{\sqrt{6}} \left(\ket{0,1,2}+\ket{1,2,0}+\ket{2,0,1}\right.
            \\
            {}&
            \left.-\ket{0,2,1}-\ket{1,0,2}-\ket{2,1,0}\right).
        \end{aligned}
        \end{equation}
        The corresponding polynomial is $ p(x,y,z) =  y z^2 + x y^2 + x^2 z - y^2 z - x z^2 - x^2 y =  (x - y)(y - z)(z - x)  $ , which can be directly translated into an optical network with the general approach discussed above.
        In contrast to the antisymmetric state, the symmetric one, where $p(x,y,z) =  y z^2 + x y^2 + x^2 z + y^2 z + x z^2 + x^2 y$, cannot be factorized, which is shown in Sec. \ref{sec:other_classes}.

\begin{figure}
    \includegraphics[width=1.\linewidth]{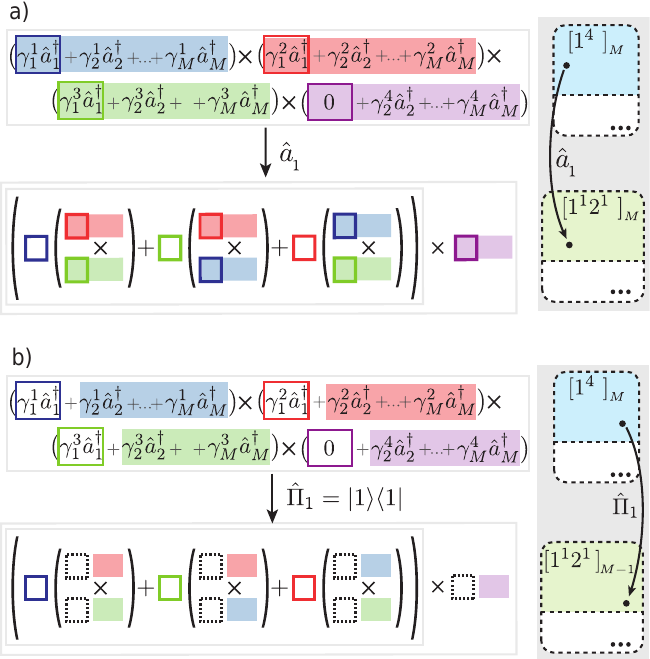} 
    \caption{%
        The structure of multivariate polynomials for the states obtained via (a) the single-photon subtraction and (b) the single-photon projection from a state belonging to the class $[1^4]_M$ with $\gamma^4_1=0$.
        Different terms are highlighted in different colors.
        The coefficients $\gamma^i_1$ are indicated by the colored squares.
        In (b), the dotted squares represent the absence of these coefficients.
        See the main text for details.
        Gray insets represent the corresponding transitions between classes under subtraction and projection, different colors represent different classes, and allowed transitions are indicated by arrows.
    }\label{fig:scheme_tensor_transformation}
\end{figure}

\begin{table*}
    \caption{
        First and second columns represent the target states and their classes $[N^1]_M$.
        The third column present the fidelities between the target and closest state that can be generated via single-photon subtraction from the class $[1^{N+1}]_M$ via the scheme depicted in Fig. \ref{fig:scheme_subtraction_projection}(a).
        The fourth column present the similar fidelities for single-photon projection from the class $[1^{N+1}]_{M+1}$ via the scheme depicted in Fig. \ref{fig:scheme_subtraction_projection}(b).
    }\label{table_results}
    \begin{ruledtabular}
    \begin{tabular}{ l c r r c }
        target state & class  & subtraction     &  projection &  \\
                    &        & fidelity         &  fidelity  &     \\
        \hline
        $\ket{ \Psi}_1 = \frac{1}{\sqrt{3}} \big(\ket{2,0,0}+\ket{0,2,0}+\ket{0,0,2} \big)$ & $[2^1]_3$ & 100.0\% & 100.0\% &   \\
        $\ket{ \Psi}_2 = \frac{1}{\sqrt{3}} \big(\ket{3,0,0}+\ket{0,3,0}+\ket{0,0,3} \big)$& $[3^1]_3$ & 83.3\% & 100.0\%  &   \\
        $\ket{ \Psi}_3 = \frac{1}{\sqrt{3}} \big(\ket{4,0,0}+\ket{0,4,0}+\ket{0,0,4} \big)$ &  $[4^1]_3$ & 66.6\% & 95.5\% &  \\
        $\ket{ \Psi}_4 = \frac{1}{\sqrt{4}} \big(\ket{2,0,0,0}+\ket{0,2,0,0}+\ket{0,0,2,0}+\ket{0,0,0,2} \big)$ &  $[2^1]_4$ & 75.0\% & 75.0\% &   \\
        $  \ket{ \Psi}_5 = \frac{1}{\sqrt{6}} \big(\ket{0,1,2}+\ket{1,2,0}+\ket{2,0,1} +\ket{0,2,1}+\ket{1,0,2}+\ket{2,1,0}\big) $ & $[3^1]_3$ &  96.4\% & 100.0\% &  \\
        $\ket{ \Psi}_6 = \frac{1}{\sqrt{3}} \big(\ket{1,1,0}+\ket{1,0,1}+\ket{0,1,1} \big)$& $[2^1]_3$ &  99.89\% &  100.0\% &  \\
        $\ket{ \Psi}_7 = \frac{1}{\sqrt{3}} \big(\ket{2,2,0}+\ket{2,0,2}+\ket{0,2,2} \big)$ & $[4^1]_3$ &  98.1\% & 100.0\%  &   \\
    \end{tabular}
    \end{ruledtabular}
\end{table*}

    \subsection{Irreducible class $[N^1]_M$}
    \label{sec:other_classes}
 
        Unlike the states in $[1^N]_M$, the generation of states from other classes needs additional nonlinearities, such as photon subtraction, measurement projections, and potentially other nonlinear processes corresponding to multiple photon generation, e.g., parametric down-conversion, four-wave mixing, etc.

        For instance, one experimentally easily accessible example of the state that is described by irreducible second-order polynomials of the class $[2^1]_M$ is the biphoton state, which may be generated via multimode parametric down-converions.
        Such a state can be represented in the form $\ket{\mathrm{\psi}} \propto \sum_i \gamma_{i,i}^2 \hat{A}^\dagger_i \hat{A}^\dagger_i\ket{\mathrm{vac}}$, where the operators $\hat{A}^\dagger_i$ represent so-called Schmidt modes \cite{PRA91}.

        In the following, let us discuss the key difference between the subtraction and projection, i.e., between the schemes in Figs. \ref{fig:scheme_subtraction_projection}(a) and \ref{fig:scheme_subtraction_projection}(b), respectively.
        The ability to generate states is related to their internal structure after subtraction and projection.
        Schematically, such a difference is depicted in Fig. \ref{fig:scheme_tensor_transformation}.
        Having the photon subtraction in the first channel, the coefficient with the subscript ``1,'' namely  $\gamma^k_{1}$, appears twice in Eq. \eqref{eq:multivar_state_subtract}:
        once as a scalar factor before the products of operators
        $\prod_{j\neq k}  \hat{c}^\dagger_j$,
        and once hidden as the coefficient $\gamma^{j \neq k}_{1}$ inside the operator $\hat{c}^\dagger_j$ itself;
        see color-highlighted squares in Fig. \ref{fig:scheme_tensor_transformation} (a).
        In contrast, in the measurement-projection scheme, the coefficient $\gamma^k_{1}$ is present only once: as a scalar factor before the term $\prod_{j\neq k}  \hat{d}^\dagger_j$.
        The other terms do not contain the coefficient $\gamma^k_{1}$ since the operators $\hat{d}^\dagger_k$ do not depend on it thanks to projection;
        see Eq. \eqref{eq:multivar_state_project} and dashed empty squares in Fig. \ref{fig:scheme_tensor_transformation}(b).

        As a result, the photon subtraction nonlinearity introduces additional constraints to the state transformation while the projection is more flexible, which is an interesting insight for future experimental implementations.
        However, this is on the expense of a reduction of available modes, which restricts scalability.

        Now, we present some states that can be generated with the schemes presented in Fig. \ref{fig:scheme_subtraction_projection}.
        In particular, for $[N^1]_M$, we give examples of states that either can or cannot be generated via photon subtraction and photon projection from the classes $[1^{N+1}]_M$ and $[1^{N+1}]_{M+1}$, respectively.

        In Table \ref{table_results}, the considered examples are presented.
        Specifically, we examine different variations of multi-mode NOON states \cite{Hong_2021} and other symmetric states.
        The third and fourth columns show the highest-found fidelity for a single subtraction and projection, respectively.
        The numerical procedure utilized to find the states is presented in Appendix \ref{appendix_numerical_procedure}.
        As we can see in Table \ref{table_results}, the three-mode NOON state $\ket{\Psi_1}$ can be generated via both the single subtraction and the projection processes with unit fidelity.
        In contrast, the state $\ket{\Psi_2}$ is only generated reliably via a single projection scheme while the states $\ket{\Psi_3}$ and $\ket{\Psi_4}$ are not reachable with a unit fidelity.
        From Table \ref{table_results}, one can see that the projection allows us to generate the states $\ket{\Psi}_2$, $\ket{\Psi}_5$, $\ket{\Psi}_6$, and $\ket{\Psi}_7$, which is impossible with a single-photon subtraction.
        This is further related to the fact that, in the case of projection, we are starting from a wider class in terms of the number of modes.
        All examples in Table \ref{table_results}, however, can be approximated with fidelities exceeding 50\%.

\begin{figure}
    \includegraphics[width=1.\linewidth]{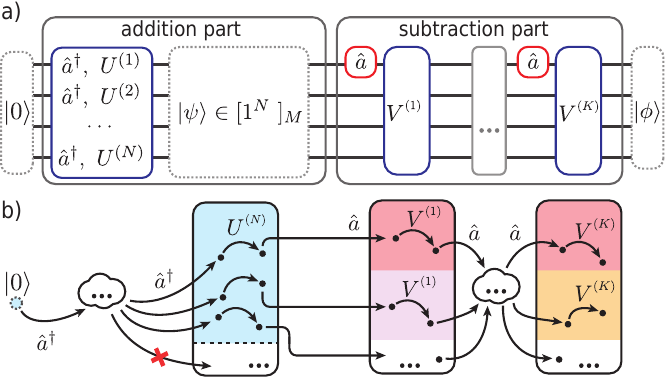} 
    \caption{%
        (a) The scheme of the multilinear interferometer with $N$ additions and rotations $U^{(n)}$ and consequent $K$-photon subtractions with rotations $V^{(m)}$.
        (b) Scheme of the transitions between states enabled by the multilayer system.
        The allowed transitions are indicated by arrows.
        The blue region corresponds to the classes $[1^N]_M$, other classes are shown in other colors.
        White regions correspond to arbitrary remaining classes.
    }\label{fig:scheme_full_additions_subtraction}
\end{figure}

    \subsection{Multiple subtractions and projections}

        We can scale the proposed schemes to encompass multiple photon subtractions, multiple projections, and the projection onto states with a higher number of photons, e.g., in a looped configuration \cite{Engelkemeier_2020}.
        The multilayer interferometer with multiple subtractions is shown in Fig. \ref{fig:scheme_full_additions_subtraction}.
        The state after multiple subtractions is obtained in a similar manner as described above by sequential actions of annihilation operators followed by unitary transformations.
        These unitary transformations increase the degrees of freedom and, together with multiple subtractions, allow us to generate states that are not reachable otherwise.

        As an example, we numerically probe whether the extended scheme up to two subtractions and rotations yield a better---possibly perfect---generation of some states discussed earlier.
        For instance, we find that the states $\ket{\Psi}_5$ and $\ket{\Psi}_6$, which were not reached via a single subtraction, can be generated by performing a unitary transformation after the single-photon subtraction.
        Moreover, the states $\ket{\Psi}_2$ and $\ket{\Psi}_4$ can be generated after two subtractions and two unitary transformations, i.e., after two layers in Fig. \ref{fig:scheme_full_additions_subtraction}.

        The proposed method of classification is based on multivariate polynomials and has a low computational cost after every intermediate layer.
        This is the basis of our numerical optimization and makes the polynomial for the final state compartively easy to compute. 
        When explicitly obtaining the final state, the main computational cost arises from finding its probability amplitudes, which, however, can be done using row vectors of polynomial coefficients of the final state.
        Note that similar  expressions can be obtained for multiple projections.
        In contrast, the approach \cite{Heurtel_2023}, for example, would require the reconstruction of the probability amplitudes after each unitary transformation in the multilayer interferometer.
        The main limitation here is the reduction in the number of modes after each projection.
        However, as mentioned before, one viable extension could be the use of the projections to $n$-photon states, with $n>1$.


\section{Conclusions}
\label{sec:conclusions}

    In this work, we established a classification scheme for multiphoton states in multimode bosonic systems.
    Equating the problem of factoring single-photon states, photon-pair states, triplet states, etc. of a multi-photon state to the problem of a factorization of multivariate polynomials led to distinct classes of states concurring with the degree of polynomial factorization.
    The base class is described via photons which may be spread across all modes but are generated individually, presenting a surprising form of global factorization property of individual photons.
    Other classes require the joint multi-photon excitation;
    for example, irreducible photon pairs corresponding to factors of quadratic polynomials that cannot be factorized linearly.
    The most advanced class contains one entirely irreducible polynomial which requires complex, highly multi-photon processes to be generated, resulting in sophisticated quantum correlations.
    Importantly, the introduced classes are shown to be invariant under arbitrary unitary mode transformations.

    Beyond the classification itself, we explored ways to generate states in each class via readily available experimental techniques.
    Our studies include arbitrary linear networks in combination with nonlinear components.
    For example, we demonstrated that all states in the base class of linearly factorizable states can be reached via photon addition and unitary transformations in a multi-layered setup.
    From this base class, we showed how photon-subtraction schemes and photon-measurement-induced nonlinearities render it possible to transition to other classes of complex, irreducible states.

    The introduced classification and relatively cheap computational costs of the methods formulated in this work is going to be useful to design optimal experiments with quantum light.
    In addition, the fabrication of quantum devices can benefit from the understanding of the state transformations available under nonlinear operations, together with the proper architecture of linear optical interferometers \cite{Reck_1994,Clements_2016,Guise_2018,Kumar_2021} and  optimization methods \cite{Krenn_2021}.

\begin{acknowledgments}
    This work is supported by the ``Photonic Quantum Computing'' (PhoQC) project, funded by the Ministry for Culture and Science of the State of North-Rhine Westphalia.
    We acknowledge financial support of the Deutsche Forschungsgemeinschaft (DFG) via the TRR 142/3 (Project No. 231447078, Subproject No. C10).
\end{acknowledgments}


\appendix 

\section{ Classification proofs }
\label{appendix_cl_proof}

\begin{lemma}
    The numbers $(n_1, n_2, ... , n_N)$ of degree 1, degree 2, $\ldots$, degree $N$ factors in Eq. \eqref{eq:factorized_state} are well-defined for any state $\Psi$.
\label{lemma_1}
\end{lemma}

\begin{proof}
    The complex polynomials $\mathbb{C}[x_1,\dots, x_N]$ constitute a unique factorization domain \cite[Thm. 2.8.12]{Sharpe_1987}.
    Thus, the representation Eq. \eqref{eq:factorized_state} is unique up to reordering the irreducible polynomials and multiplication of scalar, non-zero factors.
    Existence of the representation Eq. \eqref{eq:factorized_state} for any gives state follows because factors of homogeneous polynomials are homogeneous.
    The numbers $n_i$ correspond to the number of homogeneous irreducible polynomials of degree $i$ in Eq. \eqref{eq:factorized_state}.
    By the uniqueness properties of factorizations, the tuple $(n_1, n_2, \ldots , n_N)$ exists and is unique for any given state $\ket{\Psi}$.
\end{proof}

\begin{theorem}
    For many-mode interferometers with $M \geq 3$, the number of non-empty classes corresponds to the number of integer partitions $a(N)$ \cite[A000041]{oeis}.
\label{theorem_1}
\end{theorem}

\begin{proof}
    The homogeneous polynomial $p_k(x,y,z) = x^k+y^k-z^k$ is irreducible for any $k \in \mathbb N$ (Fermat curve).
    For any tuple $S=(n_1,\ldots,n_N)$, with $\sum_i i \times  n_i =N$, the polynomial $P =p_1^{n_1} \dots  p_N^{n_N}$ belongs to the class $[1^{n_1} \dots N^{n_N}]$.
    Thus, to any integer partition $\sum_i i \times  n_i =N$, the class $[1^{n_1} \dots N^{n_N}]$ is non-empty.
    Therefore, the number of (non-empty) equivalence classes coincides with the number of integer partitions when $M\geq 3$.
\end{proof}

\begin{theorem}
    The class $[1^{n_1} \dots N^{n_N}]$ of a state $\ket{\Psi}$ is invariant under transformations of a state $\ket{\Psi}$ by a linear interferometer.
\label{theorem_2}
\end{theorem}

\begin{proof}
    If $\ket{\Psi}$ is transformed to $\ket{\Psi'}$ by a linear interferometer, described by a unitary matrix $U$ of size $M \times M$, then the associated homogeneous polynomials $p_\Psi$ and $p_{\Psi'}$ relate as $p_{\Psi'}(\boldsymbol{x}) = p_\Psi (U \boldsymbol{x})$.
    The homogeneous polynomials $p_{\Psi'}$ and $p_{\Psi}$ have the same class $[1^{n_1} \dots N^{n_N}]$:
    Consider the polynomial $p_\Psi$ of class $[1^{n_1} \dots N^{n_N}]$ in its irreducible factorized form Eq. \eqref{eq:factorized_state}.
    Any factor constitutes a homogeneous polynomial.
    The degree of homogeneous polynomials is invariant under invertible, linear transformations.
    Thus, a factorization of $p_{\Psi'}$ with tuple $(n_1, n_2, \ldots , n_M)$ can be obtained from the irreducible factorization Eq. \eqref{eq:factorized_state} of $p_\Psi$.
    Assume one of the homogeneous factors of $p_{\Psi'}$ is not irreducible and is factorized further.
    Then, using the inverse of the considered linear transformations of variables, we achieve a factorization of $p_\Psi$ in the original variables with some tuple $(n'_1, n'_2, \ldots , n'_M)$.
    Let $k$ be the first index for which $n'_k \neq  n_k$.
    Then $n'_k > n_k$ by construction.
    Subsequent factorization applied to each factor (if necessary) yields a tuple $(n''_1, n''_2, \ldots , n''_M)$ for $p_\Psi$ different from $(n_1, n_2, \ldots , n_M)$ but this contradicts the uniqueness properties of the irreducible factorization of $p_\Psi$.
\end{proof}

\section{Non-universality of the schemes involving an interferometer with a single annihilation}
\label{appendix_SingleAnnihilationsNoGo}
\label{appendix_subt_nonunversal}

In this section, we will show that when the number of photons $N$ and the number of modes $M$ are large, a single annihilation operation cannot produce all quantum states of $N-1$ photons with $M$ modes from states of type $[1^N]_M$.
    In the following, $\mathbb{C}[x_1,\dots,x_M]_{k}$ denotes homogeneous polynomials in $\mathbb{C}[x_1,\dots,x_M]$ of degree $k$.

\begin{theorem}
    Consider the set $[1^N]_M \subset \mathbb{C}[x_1,\dots,x_M]_{N}$ of homogeneous polynomials of degree $N$ with $M$ variables that factor into a product of homogeneous polynomials of degree 1.
    Furthermore, consider the operation
    \begin{equation}
    \begin{split}
        T_\gamma \colon [1^N]_M &\to \mathbb{C}[x_1,\dots,x_M]_{N-1}\\
        \prod_{k=1}^N p_k &\mapsto \sum_{i=1}^N \gamma_i \prod_{k=1, k\not = i}^N p_k
    \end{split}
    \end{equation}
    for $\gamma \in \mathbb C^N$.
    Then the set
    \[\{T_\gamma p \; : \; \gamma \in \mathbb C^N, p \in [1^N]_M\}\]
    is a true subset of $\mathbb{C}[x_1,\dots,x_M]_{N-1}$, provided that $M,N$ are large enough, i.e.,
    \begin{equation}
        \label{equation_theorem_noGo_relation_appendix}
        MN+N - \binom{M+N-2}{M-1} < 0.
    \end{equation}
\label{proposition_for_NoGo}
\end{theorem}
 
\begin{proof}
    Any polynomial in $[1^N]_M$ can be presented in the form $p = \prod_{k=1}^N (\alpha_{1,k} x_1 + \dots + \alpha_{M,k} x_M)$ with $\alpha_{j,k} \in \mathbb C$.
    Its image under $T_\gamma$ for $\gamma \in \mathbb C^N$ has the form
    \[
        T_\gamma p = \sum_{i=1}^N\prod_{k=1, k \not = i}^N \gamma_i (\alpha_{1,k} x_1 + \dots + \alpha_{M,k} x_M).
    \]
    The expression contains $MN+N$ parameters $\alpha_{j,k},\gamma_i \in \mathbb C$.
    The space of degree $N-1$ polynomials $\mathbb{C}[x_1,\dots,x_M]_{N-1}$ has dimension
    \begin{align*}
        \dim_{\mathbb C} &(\mathbb{C}[x_1,\dots,x_M]_{N-1})\\
        &=\binom{(M-1)+(N-1)}{M-1}
        =\binom{M+N-2 }{M-1}\\
        &\sim \frac 1 {\sqrt{2\pi}}\frac{(M+N-2)^{M+N-\frac 32}}{(M-1)^{M-\frac 12}(N-1)^{N-\frac 12}}.
    \end{align*}
    Here the last relation denotes an asymptotic expansion which has been obtained from the Sterling approximation of binomial coefficients as $M,N$ both tend to infinity.
    As the dimension grows asymptotically faster than the number of parameters $MN+N$ of $T_\gamma p$ and the parameters enter in a continuously differentiable way, the set
    \[
        \{T_\gamma p \; : \; \gamma \in \mathbb C^N, p \in [1^N]_M\}
    \]
    cannot cover all of $\mathbb{C}[x_1,\dots,x_M]_{N-1}$.
\end{proof}

Theorem \ref{proposition_for_NoGo} has the following implication for a linear interferometer.

\begin{corollary}
    If
    \begin{equation}
        \label{equation_theorem_noGo_relation}
        MN+N - \binom{M+N-2}{M-1}   <0
    \end{equation}
    holds true, then a linear interferometer with a single annihilation operation cannot produce all quantum states of $N-1$ photons with $M$ modes from states of type $[1^N]_M$.
    Condition Eq. \eqref{equation_theorem_noGo_relation} is asymptotically fulfilled when the number of photons $N$ and the number of modes $M$ both become large.
\label{theorem_noGoSingleAnnihilation}
\end{corollary}

\begin{widetext}
\begin{remark}
    For concreteness, some values of
    \[
    MN+N-\dim_{\mathbb C} (\mathbb{C}[x_1,\dots,x_M]_{N-1})
    \]
    are reported in the following:
    \begin{equation}
        \begin{bmatrix}
            & M=2 &M=3&M=4&M=5&M=6&M=7\\
            N=2 & 4 & 5 & 6 & 7 & 8 & 9 \\
            N=3 & 6 & 6 & 5 & 3 & 0 & -4 \\
            N=4 & 8 & 6 & 0 & -11 & -28 & -52 \\
            N=5 & 10 & 5 & -10 & -40 & -91 & -170 \\
            N=6 & 12 & 3 & -26 & -90 & -210 & -414 \\
            N=7 & 14 & 0 & -49 & -168 & -413 & -868 \\
            N=8 & 16 & -4 & -80 & -282 & -736 & -1652 \\
            N=9 & 18 & -9 & -120 & -441 & -1224 & -2931 \\
            N=10 & 20 & -15 & -170 & -655 & -1932 & -4925 \\
            N=11 & 22 & -22 & -231 & -935 & -2926 & -7920 \\
            N=12 & 24 & -30 & -304 & -1293 & -4284 & -12280 \\
        \end{bmatrix}.
    \end{equation}
\end{remark}

    At negative values, the statement of Proposition \ref{proposition_for_NoGo} and Theorem \ref{theorem_noGoSingleAnnihilation} hold true.
    No claim is made for ranges of $(M,N)$, where values are non-negative.
\end{widetext}

\section{Numerical optimization}
\label{appendix_numerical_procedure}

    Here we present the numerical scheme for studying the desired state generation in a scheme depicted in Fig. \ref{fig:scheme_subtraction_projection}(a).
    This scheme consists of two parts:
    (i) the generation of the state $\ket{\psi}_I$ in the class $[1^{N+1}]_M$ via the photon additions and unitary transformations, and (ii) generation of a final state $\ket{\psi}_T \propto \hat{a} \ket{\psi}_I$  via the single photon subtraction.
    Therefore, the task of obtaining the target state $\ket{\psi}_T$ is reduced to the determination of proper initial state $\ket{\psi}_I$ in $[1^{N+1}]_M$, which can be done with the use of numerical multivariate optimization methods.

    Our optimization scheme is the following:
    \begin{enumerate}
        \item
        Instead of determination of matrices $U_n$ in Fig. \ref{fig:scheme_subtraction_projection}(a), we use the multiset $\Gamma$ ($N\times M$ complex parameters $\gamma^i_{j}$) to parametrize the state $\ket{\psi(\Gamma)}_I $, Eq.\eqref{eq:state_in_c}.
        \item
        The action of annihilation operator gives us the state $\ket{\psi_a(\Gamma)} = \mathcal{N} \hat{a} \ket{\psi(\Gamma)}_I$.
        \item
        As a loss function we use $\mathcal{L}(\Gamma) = 1- \mathcal{F} $, where $\mathcal{F} = |\braket{ \psi_a(\Gamma) |\psi_T}|^2$ is a fidelity function.
        \item
        For the numerical realization we use the Optax library \cite{deepmind2020jax} of the JAX framework \cite{jax2018github};
        the optimal parameters  $\gamma^i_{j}$ for a given target state $\ket{\psi}_T$ are obtained with a stochastic gradient descend based optimizer.
        \item
        So far as the gradient-based methods depend on initial conditions, in our simulations the optimization was performed for 5000 randomly generated initial conditions for multisets $\Gamma$.
    \end{enumerate}

    Note that the single subtraction does not allow to generate any desired state (see the proof in Appendix \ref{appendix_subt_nonunversal}); therefore, the loss function can be nonzero.

    For the projection scheme, Fig. \ref{fig:scheme_subtraction_projection}(b), the optimization procedure is similar to the subtraction one.
    The main difference is the intermediate state $\ket{\psi}_I$, which belongs to the class $[1^{N+1}]_{M+1}$.

\bibliography{references.bib}

\end{document}